\documentclass[11pt]{article}
\usepackage[left=1in,top=1in,right=1in,bottom=1in]{geometry}
\usepackage{times}

\usepackage{verbatim}
\usepackage{amsmath}
\usepackage{amssymb}
\usepackage{amsthm}
\usepackage{rotating}
\usepackage{algorithm}
\usepackage[noend]{algpseudocode}

\usepackage{silence}
\usepackage{tabularx}
\usepackage{graphicx}
\usepackage{url}
\usepackage{multirow}
\usepackage{subfig}

\usepackage[noend]{algpseudocode}
\usepackage{siunitx}

\newtheorem{theorem}{Theorem}
\newtheorem{definition}{Definition}
\newtheorem{lemma}{Lemma}
\newtheorem{proposition}{Proposition}
\newtheorem{corollary}{Corollary}
\DeclareMathOperator*{\argmax}{arg\,max}

\newcommand{\Tau}{\mathrm{T}}

\allowdisplaybreaks

\begin{document}
\title{Observable Perfect Equilibrium}
\author{Sam Ganzfried\\
Ganzfried Research\\
sam.ganzfried@gmail.com
}

\date{\vspace{-5ex}}

\maketitle

\begin{abstract}
While Nash equilibrium has emerged as the central game-theoretic solution concept, many important games contain several Nash equilibria and we must determine how to select between them in order to create real strategic agents. Several Nash equilibrium refinement concepts have been proposed and studied for sequential imperfect-information games, the most prominent being trembling-hand perfect equilibrium, quasi-perfect equilibrium, proper equilibrium, and recently one-sided quasi-perfect equilibrium. These concepts are robust to certain arbitrarily small mistakes, and are guaranteed to always exist; however, we argue that none of these is the correct concept for developing strong agents in sequential games of imperfect information. We define a new equilibrium refinement concept for extensive-form games called observable perfect equilibrium in which the solution is robust over trembles in publicly-observable action probabilities (not necessarily over all action probabilities that may not be observable by opposing players). Observable perfect equilibrium correctly captures the assumption that the opponent is playing as rationally as possible given mistakes that have been observed (while previous solution concepts do not). We prove that observable perfect equilibrium is always guaranteed to exist, and demonstrate that it selects a unique equilibrium in
a simplified no-limit poker game where the prior extensive-form refinements retain a continuum of equilibria.
We expect observable perfect equilibrium to be a useful equilibrium refinement concept for modeling many important imperfect-information games of interest in artificial intelligence. 
\end{abstract}

\section{Introduction}
\label{se:intro}
When developing a strategy for a human or computer agent to play in a game, the starting point is typically
a Nash equilibrium. Even if additional information is available about the opponents, e.g., from historical data or observations of play, we would often still opt to start playing a Nash equilibrium strategy until we are confident in our ability to successfully exploit opponents by deviating~\cite{Ganzfried11:Game,Ganzfried15:Safe}. It is well known that several conceptual and computational limitations exist for Nash equilibrium. For multiplayer and two-player non-zero-sum games, it is PPAD-hard to compute or approximate one Nash equilibrium~\cite{Chen05:Nash,Chen06:Settling,Daskalakis09:Complexity,Rubinstein17:Settling}, different Nash equilibria may give different values to the players, and following a Nash equilibrium strategy provides no performance guarantee. Even for two-player zero-sum games, in which these issues do not arise, there can still exist multiple Nash equilibria that we must select from. Therefore several solution concepts that refine Nash equilibrium in various ways have been proposed to help select one that is more preferable in some way. Most of the common equilibrium refinements are based on the idea of ensuring robustness against certain arbitrarily small ``trembles'' in players' execution of a given strategy.  Variants of these Nash equilibrium refinements have been devised for simultaneous normal-form games as well as sequential games of perfect and imperfect information. In this paper we will be primarily interested in sequential games of imperfect information, which are more complex than the other games classes and have received significant interest recently in artificial intelligence due to their ability to model many important scenarios. To simplify analysis we will primarily be studying a subclass of these games in which there are two players, only one player has private information, and both players take a single action; however, our results apply broadly to extensive-form imperfect-information games. We will also be primarily focused on two-player zero-sum games, though some analysis also applies to two-player non-zero-sum and multiplayer games. We will show that existing Nash equilibrium refinement concepts have limitations in sequential imperfect-information games, and propose the new concept of observable perfect equilibrium that addresses these limitations.

A \emph{normal-form game} consists of a finite set of players $N = \{1,\ldots,n\}$, a finite set of pure strategies $S_i$ for each player $i \in N$, and a real-valued utility for each player for each strategy vector (aka \emph{strategy profile}), $u_i : \times_i S_i \rightarrow \mathbb{R}$. A \emph{mixed strategy} $\sigma_i$ for player $i$ is a probability distribution over pure strategies, where $\sigma_i(s_{i'})$ is the probability that player $i$ plays pure strategy $s_{i'} \in S_i$ under $\sigma_i$. Let $\Sigma_i$ denote the full set of mixed strategies for player $i$. A strategy profile $\sigma^* = (\sigma^*_1,\ldots,\sigma^*_n)$ is a \emph{Nash equilibrium} if $u_i(\sigma^*_i,\sigma^*_{-i}) \geq u_i(\sigma_i, \sigma^*_{-i})$ for all $\sigma_i \in \Sigma_i$ for all $i \in N$, where $\sigma^*_{-i} \in \Sigma_{-i}$ denotes the vector of the components of strategy $\sigma^*$ for all players excluding player $i$. Here $u_i$ denotes the expected utility for player $i$, and $\Sigma_{-i}$ denotes the set of strategy profiles for all players excluding player $i$. 

Nash equilibrium has emerged as the central solution concept in game theory, and is guaranteed to exist in all finite normal-form games~\cite{Nash50:Eq,Nash51:Non}. However, games may contain multiple Nash equilibria and it is not clear which one should be played. A popular refinement of Nash equilibrium is \emph{trembling hand perfect equilibrium}~\cite{Selten75:Reexamination}. Given a normal-form game $G$, a \emph{perturbed game}
$G(\eta)$ restricts each player $i$ to mixed strategies satisfying
$\sigma_i(s) \geq \eta_{i,s}$ for every $s \in S_i$, where
$\eta_{i,s}>0$ and $\sum_{s\in S_i}\eta_{i,s}<1$.
A strategy profile $\sigma^*$ is a trembling-hand perfect
equilibrium (THPE) if there exist a sequence of perturbations
$\eta^{(k)}$ converging componentwise to zero and Nash equilibria
$\sigma^{(k)}$ of $G(\eta^{(k)})$ such that
$\sigma^{(k)} \rightarrow \sigma^*$~\cite{vanDamme83:Refinements}.
The lower bounds may differ across players and pure strategies
and need not converge to zero at the same rate.
It has been shown that every finite normal-form game has at least one THPE~\cite{Selten75:Reexamination}. The following result provides an alternative equivalent characterization of THPE~\cite{vanDamme87:Stability}:

\begin{theorem}
\label{th:thpe}
Let $\sigma^*$ be a strategy profile of an $n$-player normal-form game $G$. Then $\sigma^*$ is a trembling-hand perfect equilibrium if and only if $\sigma^*$ is a limit point of a sequence $\{\sigma(\epsilon)\}_{\epsilon \rightarrow 0}$ of totally mixed strategy profiles with the property that $\sigma^*$ is a best response for all players against every element $\sigma(\epsilon)$ in this sequence. 
\end{theorem}

While the normal form can be used to model simultaneous actions, settings with sequential moves are typically modelled using the \emph{extensive form} representation. The extensive form can also model simultaneous actions, as well as chance events and imperfect information (i.e., situations where some information is available to only some of the agents and not to others). Extensive-form games consist primarily of a game tree; each non-terminal node has an associated player (possibly \emph{chance}) that makes the decision at that node, and each terminal node has associated utilities for the players.  Additionally, game states are partitioned into \emph{information sets}, where the player whose turn it is to move cannot distinguish among the states in the same information set.  Therefore, in any given information set, a player must choose actions with the same distribution at each state contained in the information set. If no player forgets information that he previously knew, we say that the game has \emph{perfect recall}. A (behavioral) \emph{strategy} for player $i,$ $\sigma_i \in \Sigma_i,$ is a function that assigns a probability distribution over all actions at each information set belonging to $i$. Similarly to normal-form games, it can be shown that all extensive-form games with perfect recall contain at least one Nash equilibrium in mixed strategies. The concept of \emph{extensive-form trembling-hand perfect equilibrium} (EFTHPE) is defined analogously by imposing positive
lower bounds on behavioral action probabilities. At each
information set $I$ belonging to player $i$, let $A(I)$ denote
the available actions and require
$\sigma_i(a\mid I)\geq\eta_{i,I,a}>0$ for every $a\in A(I)$,
where $\sum_{a\in A(I)}\eta_{i,I,a}<1$.
An EFTHPE is a limit of Nash equilibria of these perturbed games
as all lower bounds converge to zero. The bounds may differ
across actions, including actions at the same information set,
and need not converge to zero at the same rate.
Every finite extensive-form game with perfect recall has at
least one EFTHPE~\cite{vanDamme87:Stability}.

The concept of \emph{extensive-form proper equilibrium}
refines EFTHPE: every extensive-form proper equilibrium
is an EFTHPE, but the converse need not hold.
For a completely mixed behavioral strategy profile $\sigma$,
let $V_i(I,a;\sigma)$ denote player $i$'s expected payoff
conditional on reaching information set $I$ and choosing
action $a$, with all other action probabilities unchanged.
The profile is \emph{extensive-form $\epsilon$-proper} if,
at every information set $I$ belonging to player $i$,
\[
V_i(I,a;\sigma)<V_i(I,b;\sigma)
\quad\Longrightarrow\quad
\sigma_i(a\mid I)\leq\epsilon\,\sigma_i(b\mid I)
\]
for all $a,b\in A(I)$.
An extensive-form proper equilibrium is a limit of
extensive-form $\epsilon$-proper profiles as
$\epsilon\rightarrow 0$.
Equivalently, it is a proper equilibrium of the
\emph{agent normal form}, in which each information set
is controlled by a separate agent sharing the payoff
of the original player~\cite{vanDamme83:Refinements}.

\section{Observable Perfect Equilibrium}
\label{se:ope}
In order to simplify our analysis we define a subset of extensive-form information games called \emph{two-player one-step extensive-form imperfect-information games} (OSEFGs):
\begin{itemize}
\item There are two players, P1 and P2.
\item Player 1 is dealt private information $\tau_1$ from a finite set $\Tau_1$ uniformly at random.\footnote{Note that all of our analysis will still hold if we assume that $\tau_1$ is selected from an arbitrary probability distribution.}
\item Player 1 can then choose action $a_1$ from finite set $A_1.$
\item Player 2 observes the action $a_1$ but not $\tau_1.$
\item Player 2 then chooses action $a_2$ from finite set $A_2.$
\item Both players are then given payoff $u_i(\tau_1,a_1,a_2).$ 
\end{itemize}

For mixed strategy $\sigma_1$ for player 1, $\sigma_1(\tau_1,a_1)$ denotes the probability that player 1 takes action $a_1 \in A_1$ with private information $\tau_1 \in \Tau_1.$ Similarly for mixed strategy $\sigma_2$ for player 2, $\sigma_2(a_1,a_2)$ denotes the probability that player 2 takes action $a_2 \in A_2$ following observed action $a_1 \in A_1$ of player 1. 

Suppose we are in the position of player 2 responding to the observed action $a_1 \in A_1.$ If both players are following a Nash equilibrium strategy, then we know that we are best responding to player 1's strategy. However, suppose that we are following our component from Nash equilibrium strategy profile $\sigma^*$ in which $\sigma^*_1(\tau_1,a_1) = 0$ for all $\tau_1 \in \Tau_1.$  Our observation is clearly inconsistent with player 1 following $\sigma^*,$ since they would never choose action $a_1.$ Since our strategy is part of a Nash equilibrium, it ensures that player 1 cannot profitably deviate from $\sigma^*_1$ with any $\tau_1 \in \Tau_1$ and take $a_1$; however, there may be many such strategies, and we would like to choose the best one given that we have actually observed player 1 irrationally playing $a_1.$ In this situation, playing an EFTHPE strategy may ensure that we play a stronger strategy against this opponent, who has selected an action that they should not rationally play, since EFTHPE explicitly ensures robustness against the possibility of ``trembling'' and playing such an action with small probability. 

Extensive-form trembling-hand perfect equilibrium assumes that all players take all actions at all information sets with nonzero probability. In the situation described above, we know that player 1 is taking $a_1$ at some information set with nonzero probability; however, we really have no further information beyond that. It is very possible that they are playing a strategy that takes $a'_1 \in A_1$ with zero probability with all $\tau_1 \in \Tau_1.$ The core assumption of game theory is that, in the absence of any information the contrary, we assume that all players are behaving rationally. Now clearly that assumption is violated in this case when we observe player 1 irrationally playing $a_1.$ However, it seems a bit extreme to now assume that all players are playing all actions with nonzero probability. If we assume that the opponent is playing \emph{as rationally as possible given our observations}, then we would only consider trembles that are consistent with our observations of their play. Such trembles must satisfy $\sigma_1(\tau_1,a_1) > 0$ for at least one $\tau_1 \in \Tau_1$, or alternatively, $\sum_{\tau_1 \in \Tau_1}\sigma_1(\tau_1,a_1) > 0.$ The concept of \emph{observable perfect equilibrium} (OPE) captures this assumption that all players are playing as rationally as possible subject to the constraint that their play is consistent with our observations.

\begin{definition}
\label{de:ope}
Let $G$ be a two-player one-step extensive-form
imperfect-information game, and suppose that player 2
has observed public action $a_1$ from player 1.
For $0<\epsilon<1$, let $G_\epsilon$ be obtained from
$G$ by imposing the single additional constraint
\[
\sum_{\tau_1\in\Tau_1}\sigma_1(\tau_1,a_1)
\geq\epsilon
\]
on player 1's strategy.
A Nash equilibrium of $G_\epsilon$ is called an
$\epsilon$-observable-perfect strategy profile.
A strategy profile $\sigma^*$ is an
\emph{observable perfect equilibrium}, relative to
the observed action $a_1$, if there exist a sequence
$\epsilon_k\in(0,1)$ converging to zero and Nash
equilibria $\sigma^{(k)}$ of $G_{\epsilon_k}$
such that $\sigma^{(k)}\rightarrow\sigma^*$.
\end{definition}

\begin{proposition}
Every observable perfect equilibrium is a Nash equilibrium.
\end{proposition}

\begin{proof}
Let $\sigma^*$ be an observable perfect equilibrium
relative to the observed action $a_1$. By definition,
there exist $\epsilon_k\rightarrow 0$ and Nash equilibria
$\sigma^{(k)}$ of $G_{\epsilon_k}$ such that
$\sigma^{(k)}\rightarrow\sigma^*$.

Every player-2 strategy $\widehat{\sigma}_2$ remains
feasible in every perturbed game. Therefore,
\[
u_2(\sigma^{(k)}_1,\sigma^{(k)}_2)
\geq
u_2(\sigma^{(k)}_1,\widehat{\sigma}_2).
\]
Taking limits and using continuity shows that player 2
cannot profitably deviate from $\sigma^*$.

Now fix any player-1 strategy $\widehat{\sigma}_1$.
Let $q_1$ choose $a_1$ with probability 1 for every
private type, and define
\[
\widehat{\sigma}^{(k)}_1
=(1-\epsilon_k)\widehat{\sigma}_1+\epsilon_k q_1.
\]
This strategy converges to $\widehat{\sigma}_1$ and
satisfies
\[
\sum_{\tau_1\in\Tau_1}
\widehat{\sigma}^{(k)}_1(\tau_1,a_1)
\geq \epsilon_k|\Tau_1|
\geq \epsilon_k.
\]
It is therefore feasible in $G_{\epsilon_k}$.
Since $\sigma^{(k)}$ is a Nash equilibrium of that game,
\[
u_1(\sigma^{(k)}_1,\sigma^{(k)}_2)
\geq
u_1(\widehat{\sigma}^{(k)}_1,\sigma^{(k)}_2).
\]
Taking limits gives
\[
u_1(\sigma^*_1,\sigma^*_2)
\geq
u_1(\widehat{\sigma}_1,\sigma^*_2).
\]
Thus neither player can profitably deviate, and
$\sigma^*$ is a Nash equilibrium.
\end{proof}

We can extend Definition~\ref{de:ope} to general $n$-player extensive-form imperfect-information games by adding analogous constraints for all observed actions (i.e., by requiring that the sum of the probabilities of strategy sequences taken consistent with our observations is nonzero). This is useful because we are no longer required to reason about trembles that are incompatible with the current path of play, which are irrelevant at this point. For general extensive-form games, there is a further consideration about what trembles should be considered for future moves beyond the current path of play. Extensive-form trembling hand perfect equilibrium assumes that all players may tremble in future actions, while an alternative concept called \emph{quasi-perfect equilibrium} (QPE) assumes that only the opposing players tremble for future actions (even if we have trembled previously ourselves)~\cite{vanDamme84:Relation}. The related concept of \emph{one-sided quasi-perfect equilibrium} (OSQPE) assumes that only the opposing players can tremble at all and we cannot~\cite{Farina21:Equilibrium}. OSQPE is the most computationally efficient and is also the most similar to observable perfect equilibrium.

In two-player one-step extensive-form imperfect-information games QPE and EFTHPE are identical, since both players only take a single action along the path of play. They both require that both player 1 and player 2 put nonzero probability trembles on all actions. Both of these potentially differ from OSQPE, which from the perspective of player 2 requires that player 1 puts nonzero probability trembles on all possible actions, while player 2 does not have this requirement. Note that all three of these concepts potentially differ from observable perfect equilibrium (from the perspective of player 2), which requires only that player 1 makes a nonzero probability tremble at some information set consistent with taking the observed action $a_1$ and no requirement on player 2.

Observable perfect equilibrium is fundamentally different from other equilibrium refinements in that it is dependent on the action taken by the other player. For each observed action $a_1$ we compute a potentially different strategy for ourselves, conditionally on having observed $a_1.$ In contrast, all other equilibrium refinement concepts compute a full strategy profile for all players in advance of game play. If play contains longer sequences of actions than just a single move for each player, we can recompute our OPE strategy after each new observation (again by assuming positive probability trembles for all actions consistent with the path of play). For each fixed observation, an OPE defines a full-game
Nash equilibrium profile, whose prescribed response
can be used at the information set reached.

We can view the relation between OPE and other solution concepts analogously to the relation between endgame solving~\cite{Ganzfried15:Endgame} and standard offline game solving in large imperfect-information games. Previously the standard approach for approximating Nash equilibrium strategies in large imperfect-information games was to first apply an abstraction algorithm to create a significantly smaller game that is strategically similar to the full game~\cite{Shi01:Abstraction,Billings03:Approximating,Gilpin06:Competitive,Gilpin07:Better,Gilpin08:Heads-up,Waugh09:Practical,Johanson13:Evaluating}, then solve the abstract game using an equilibrium-finding algorithm such as counterfactual regret minimization (CFR)~\cite{Zinkevich07:Regret} or a generalization of Nesterov's excessive gap technique~\cite{Hoda10:Smoothing}. With endgame solving~\cite{Ganzfried15:Endgame}, the portion of the game tree that we have reached is solved in real time to a finer degree of granularity than in the offline abstract equilibrium computation. This focused computation led to superhuman play in two-player~\cite{Brown17:Superhuman,Moravcik17b:DeepStack} and six-player no-limit Texas hold 'em~\cite{Brown19:Superhuman}. OPE similarly offers potential computational advantages
over other equilibrium refinement concepts such as EFTHPE,
QPE, and OSQPE, by focusing the perturbation on the
opponent's observed play rather than requiring trembles
on all actions throughout the game.
However, unlike endgame solving, OPE still guarantees that the computed strategy remains a Nash equilibrium, while it has been shown that endgame solving may produce strategies that are not Nash equilibria of the full game (even if the trunk strategy were an exact Nash equilibrium)~\cite{Ganzfried15:Endgame}. 

In order to show existence of observable perfect equilibrium and an algorithm for its computation, we first review results for the related concept of one-sided quasi-perfect equilibrium. One-sided quasi-perfect equilibrium assumes that the game is two-player zero-sum and that we play the role of the ``machine player'' while the opponent is the ``human player.'' A key subroutine in the computation of one-sided quasi-perfect equilibrium in two-player zero-sum games is the solution to an optimization formulation for an $\epsilon$-quasi-perfect equilibrium strategy profile (where the trembles for the human player are lower bounded by $\boldsymbol{\ell}_h(\epsilon)$). The vector $\boldsymbol{\ell}_h(\epsilon)$ has entries $\boldsymbol{\ell}_h(\epsilon)[\sigma] = \epsilon^{|\sigma|}$, where $|\sigma|$ denotes the number of actions for player $h$ in the sequence $\sigma.$ In the OSEFG setting there is just a single action per player per sequence, so we have $\ell_h(\epsilon) = \epsilon.$ It has been shown that the optimization formulation (\ref{eq:OSQPE}) corresponds to a one-sided quasi-perfect equilibrium~\cite{Farina21:Equilibrium}. In (\ref{eq:OSQPE}), $\mathbf{x}_m$ is the vector of the machine player's strategy sequences, $\mathbf{x}_h$ is the vector of the human player's strategy sequences, and $\mathbf{A}_m$ is the payoff matrix from the machine player's perspective. The vector $\mathbf{f}_h$ and matrix $\mathbf{F}_h$ are constants that encode the sequence-form representation of the information set structure for the human player, and $\mathbf{f}_m,\mathbf{F}_m$ are analogous for the machine player~\cite{Koller94:Fast}. 

\begin{equation} \label{eq:OSQPE}
\begin{array}{rrl} 
&\max_{\mathbf{x}_m} \min_{\mathbf{x}_h} & \mathbf{x}^T_m \mathbf{A}_m \mathbf{x}_h \\
&\mbox{s.t.}& \mathbf{F}_m \mathbf{x}_m =  \mathbf{f}_m \\
& & \mathbf{x}_m \geq \mathbf{0} \\
& & \mathbf{F}_h  \mathbf{x}_h =  \mathbf{f}_h \\
& & \mathbf{x}_h \geq \boldsymbol{\ell}_h(\epsilon) \\
\end{array}
\end{equation}

The following result has been proven from this formulation~\cite{Farina21:Equilibrium}.
\begin{lemma}
Consider the bilinear saddle point problem (\ref{eq:OSQPE}). Then, for any $\epsilon > 0$ for which the domain of the minimization problem is nonempty,
any solution to (\ref{eq:OSQPE}) is a one-sided $\epsilon$-quasi-perfect strategy profile.
\label{le:OSQPE}
\end{lemma}

From Lemma~\ref{le:OSQPE} they are able to prove the existence of one-sided quasi-perfect equilibrium as a corollary~\cite{Farina21:Equilibrium}.

\begin{corollary}
Every two-player zero-sum extensive-form game with perfect recall has at least one one-sided quasi-perfect equilibrium.
\label{co:OSQPE}
\end{corollary}

We can obtain analogous results for observable perfect equilibrium using similar reasoning. We refer to the players as we did originally, where player 2 corresponds to ourselves (i.e., the ``machine player''), and player 1 corresponds to the opponent (i.e., the ``human player''). Note that the optimization formulation (\ref{eq:OPE}) differs from (\ref{eq:OSQPE}) only in the constraints for player 1's strategy vector. The inequality $\mathbf{c}^T \mathbf{x}_1 \geq \epsilon$ encodes the constraint that the sum of the probabilities that player 1 takes strategy sequences that are consistent with our observations of play so far is at least $\epsilon,$ where $\mathbf{c}$ is a constant vector.

\begin{equation} \label{eq:OPE}
\begin{array}{rrl} 
&\max_{\mathbf{x}_2} \min_{\mathbf{x}_1} & \mathbf{x}^T_2 \mathbf{A}_2 \mathbf{x}_1 \\
&\mbox{s.t.}& \mathbf{F}_2 \mathbf{x}_2 =  \mathbf{f}_2 \\
& & \mathbf{x}_2 \geq \mathbf{0} \\
& & \mathbf{F}_1  \mathbf{x}_1 =  \mathbf{f}_1 \\
& & \mathbf{x}_1 \geq \mathbf{0} \\
& & \mathbf{c}^T \mathbf{x}_1 \geq \epsilon \\
\end{array}
\end{equation}

\begin{lemma}
Consider the bilinear saddle point problem (\ref{eq:OPE}). Then, for any $\epsilon > 0$ for which the domain of the minimization problem is nonempty,
any solution to (\ref{eq:OPE}) is an $\epsilon$-observable-perfect strategy profile.
\label{le:OPE}
\end{lemma}

\begin{corollary}
Every two-player zero-sum extensive-form game with perfect recall has at least one observable perfect equilibrium.
\label{co:OPE}
\end{corollary}

For one-sided quasi-perfect equilibrium, it has been shown that formulation (\ref{eq:OSQPE}) implies the following linear programming formulation (\ref{eq:OSQPE-2}) for an $\epsilon$-quasi-perfect equilibrium strategy for the machine player~\cite{Farina21:Equilibrium}. In (\ref{eq:OSQPE-2}), $\mathbf{v}$ is a new vector of free variables, while all other quantities are the same as in (\ref{eq:OSQPE}).

\begin{equation} \label{eq:OSQPE-2}
\begin{array}{rrl} 
&\argmax_{\mathbf{x}_m,\mathbf{v}} & (\mathbf{A}_m \boldsymbol{\ell}_h(\epsilon))^T \mathbf{x}_m + (\mathbf{f}_h - \mathbf{F}_h \boldsymbol{\ell}_h(\epsilon))^T \mathbf{v} \\
&\mbox{s.t.}& \mathbf{A}^T_m \mathbf{x}_m - \mathbf{F}_h^T \mathbf{v} \geq \mathbf{0} \\ 
& &\mathbf{F}_m \mathbf{x}_m =  \mathbf{f}_m \\
& & \mathbf{x_m} \geq \mathbf{0} \\
& & \mathbf{v} \mbox{ free} \\
\end{array}
\end{equation}

\begin{proposition}
\label{pr:osqpe}
Any limit point of solutions to the trembling linear program (\ref{eq:OSQPE-2}) as the trembling magnitude $\epsilon \rightarrow 0^+$ is a one-sided quasi-perfect equilibrium strategy for the machine player.
\end{proposition}

This linear programming formulation leads to a polynomial-time algorithm for computation of one-sided quasi-perfect equilibrium by solving the problem for consecutively smaller values of $\epsilon$ until a termination criterion is met~\cite{Farina21:Equilibrium}. It is argued that this computation is more efficient than analogous algorithms for extensive-form trembling-hand perfect equilibrium and quasi-perfect equilibrium, because the OSQPE formulation only depends on $\epsilon$ through the objective, while EFTHPE depends on $\epsilon$ through the left-hand-side of constraints, and QPE depends on $\epsilon$ in both the right-hand-side of constraints and the objective. We now provide a linear program formulation for \(\epsilon\)-observable-perfect equilibrium that is analogous to (\ref{eq:OSQPE-2}), which also implies a polynomial-time algorithm for computation of observable perfect equilibrium with the same computational advantages as for OSQPE. Note that in our new LP formulation (\ref{eq:OPE-2}) \(\epsilon\) occurs only in the objective, as it did in (\ref{eq:OSQPE-2}). Note also that our objective omits the term involving the machine-player strategy vector that appears in the corresponding OSQPE objective, replacing the tremble contribution by the single scalar term \(\epsilon w\), which suggests that OPE can be computed faster than OSQPE. This agrees with the intuition described above, since OPE only considers trembles for the opponent consistent with our observations of the path of play, while OSQPE considers all possible past and future trembles for the opponent.

\begin{equation}
\label{eq:OPE-2}
\begin{aligned}
\operatorname*{arg\,max}_{\mathbf{x}_2,w,\mathbf{v}}\quad
& \mathbf{f}_1^T \mathbf{v} + \epsilon w \\
\text{s.t.}\quad
& \mathbf{F}_1^T \mathbf{v} + w\mathbf{c} \le \mathbf{A}_2^T \mathbf{x}_2, \\
& \mathbf{F}_2\mathbf{x}_2 = \mathbf{f}_2, \\
& \mathbf{x}_2 \ge \mathbf{0}, \\
& w \ge 0, \\
& \mathbf{v} \text{ free.}
\end{aligned}
\end{equation}

\begin{proposition}
Any limit point of solutions to the trembling linear program \((\ref{eq:OPE-2})\), as the
trembling magnitude \(\epsilon\) goes to zero, is an observable perfect
equilibrium strategy for player 2.
\end{proposition}

\begin{proof}
Fix \(\epsilon > 0\) such that the feasible region of the inner minimization
problem in \((\ref{eq:OPE})\) is nonempty. For a fixed player-2 sequence-form strategy
\(\mathbf{x}_2\), the inner minimization problem in \((\ref{eq:OPE})\) is
\begin{equation*}
\begin{aligned}
\min_{\mathbf{x}_1}\quad
& \mathbf{x}_2^T \mathbf{A}_2 \mathbf{x}_1 \\
\text{s.t.}\quad
& \mathbf{F}_1\mathbf{x}_1 = \mathbf{f}_1, \\
& \mathbf{c}^T \mathbf{x}_1 \ge \epsilon, \\
& \mathbf{x}_1 \ge \mathbf{0}.
\end{aligned}
\end{equation*}
Since
\(\mathbf{x}_2^T \mathbf{A}_2 \mathbf{x}_1
= (\mathbf{A}_2^T \mathbf{x}_2)^T \mathbf{x}_1\), this is equivalently
\begin{equation*}
\begin{aligned}
\min_{\mathbf{x}_1}\quad
& (\mathbf{A}_2^T \mathbf{x}_2)^T \mathbf{x}_1 \\
\text{s.t.}\quad
& \mathbf{F}_1\mathbf{x}_1 = \mathbf{f}_1, \\
& \mathbf{c}^T \mathbf{x}_1 \ge \epsilon, \\
& \mathbf{x}_1 \ge \mathbf{0}.
\end{aligned}
\end{equation*}
The dual of this linear program is
\begin{equation}
\label{eq:OPE-dual}
\begin{aligned}
\max_{\mathbf{v},w}\quad
& \mathbf{f}_1^T \mathbf{v} + \epsilon w \\
\text{s.t.}\quad
& \mathbf{F}_1^T \mathbf{v} + w\mathbf{c} \le \mathbf{A}_2^T \mathbf{x}_2, \\
& w \ge 0, \\
& \mathbf{v} \text{ free.}
\end{aligned}
\end{equation}
The dual constraint is an inequality because the primal variable
\(\mathbf{x}_1\) is constrained to be nonnegative. Since the primal feasible
region is nonempty and contained in the sequence-form strategy space, the
primal is feasible and bounded. Thus strong duality applies.

Substituting the dual problem \((\ref{eq:OPE-dual})\) for the inner minimization problem in
\((\ref{eq:OPE})\), the problem of computing a player-2 strategy for the perturbed
observable-perfect problem is equivalent to the linear program
\begin{equation*}
\begin{aligned}
\operatorname*{arg\,max}_{\mathbf{x}_2,w,\mathbf{v}}\quad
& \mathbf{f}_1^T \mathbf{v} + \epsilon w \\
\text{s.t.}\quad
& \mathbf{F}_1^T \mathbf{v} + w\mathbf{c} \le \mathbf{A}_2^T \mathbf{x}_2, \\
& \mathbf{F}_2\mathbf{x}_2 = \mathbf{f}_2, \\
& \mathbf{x}_2 \ge \mathbf{0}, \\
& w \ge 0, \\
& \mathbf{v} \text{ free.}
\end{aligned}
\end{equation*}
This is precisely the trembling linear program \((\ref{eq:OPE-2})\). Hence any
solution \(\mathbf{x}_2^\epsilon\) of \((\ref{eq:OPE-2})\) is a player-2 solution of the
perturbed observable-perfect problem \((\ref{eq:OPE})\).

By Lemma \ref{le:OPE}, for every \(\epsilon > 0\) for which the feasible region of the
inner minimization problem is nonempty, any solution to \((\ref{eq:OPE})\) is an
\(\epsilon\)-observable-perfect strategy profile. Thus each \(\mathbf{x}_2^\epsilon\)
obtained from \((\ref{eq:OPE-2})\) is the player-2 component of an
\(\epsilon\)-observable-perfect strategy profile.

Now let \(\{\epsilon_k\}_{k=1}^{\infty}\) be any sequence with
\(\epsilon_k > 0\) for all \(k\) and \(\epsilon_k \to 0\). For each \(k\), let
\(\mathbf{x}_2^{\epsilon_k}\) be a solution of \((\ref{eq:OPE-2})\) with trembling
magnitude \(\epsilon_k\). Since the player-2 sequence-form strategy space is
compact, the sequence \(\{\mathbf{x}_2^{\epsilon_k}\}_{k=1}^{\infty}\) has a
convergent subsequence. Let \(\mathbf{x}_2^*\) be any limit point of this
sequence. Since each \(\mathbf{x}_2^{\epsilon_k}\) is the player-2 component
of an \(\epsilon_k\)-observable-perfect strategy profile and
\(\epsilon_k \to 0\), the limit point \(\mathbf{x}_2^*\) is an observable
perfect equilibrium strategy for player 2.
\end{proof}

For games with rational payoffs and chance probabilities,
formulation~(\ref{eq:OPE-2}) has polynomial encoding size
and coefficients that are rational polynomials in $\epsilon$.
After converting the LP to standard form, its constraint
matrix and right-hand side are independent of $\epsilon$.
Cramer's rule therefore gives a common bound of polynomial
encoding length on all basic feasible solutions.
Imposing this bound on the standard-form variables
preserves an optimal solution for every sufficiently small
positive $\epsilon$ and yields a bounded feasible region.
It therefore falls within the framework of trembling
linear programs studied by Farina et al.,
whose algorithm computes a limit solution in polynomial
time~\cite{Farina18:Practical}.
Together with the preceding proposition, this establishes
polynomial-time computation of an observable perfect
equilibrium strategy for player 2.

\section{No-limit poker}
\label{se:poker}
In this section we show that observable perfect equilibrium
selects a unique equilibrium in a simplified no-limit poker
game where the other extensive-form refinements considered
here retain the full set of Nash equilibria.
Poker has been a major AI challenge problem in recent years, with no-limit Texas hold 'em in particular being the most popular variant for humans. No-limit Texas hold 'em is a large sequential game of imperfect information, and just recently computers have achieved superhuman performance, first in the two-player variant~\cite{Moravcik17b:DeepStack,Brown17:Superhuman} and subsequently for six players~\cite{Brown19:Superhuman}. These agents attempt to compute approximations of Nash equilibrium strategies by first running an \emph{abstraction algorithm} to create a smaller strategically-similar game, and then solving the abstract game using an equilibrium-finding algorithm such as counterfactual regret minimization. Counterfactual regret minimization (CFR) is an iterative self-play procedure that has been proven to converge to Nash equilibrium in two-player zero-sum~\cite{Zinkevich07:Regret}, though it has been demonstrated to not converge to Nash equilibrium in a simplified three-player poker game~\cite{Abou10:Using}. The key insight that led to superhuman play was to combine these abstraction and equilibrium-finding approaches with endgame solving~\cite{Ganzfried15:Endgame}, in which the portion of the game we have reached during a hand is resolved in a finer granularity in real time. It is somewhat remarkable that these approaches have achieved such strong performance despite numerous theoretical limitations: the abstraction algorithms have no performance guarantee, endgame solving can produce strategies that are not Nash equilibria of the full game, CFR does not guarantee convergence to Nash equilibrium for more than two players, and furthermore even if CFR did converge for more than two players there can be multiple Nash equilibria and following one has no performance guarantee. It turns out that even ignoring all of these theoretical limitations, there is an additional challenge present. Even if we are in the two-player zero-sum setting and are able to compute an exact Nash equilibrium, the game may contain many Nash equilibria, and we would like to choose the ``best'' one. As we will see, even the simplest two-player no-limit poker game contains infinitely many Nash equilibria. 

In the \emph{no-limit clairvoyance game}~\cite{Ankenman06:Mathematics}, player 1 is dealt a \emph{winning hand} (W) and a \emph{losing hand} (L) each with probability $\frac{1}{2}.$ (While player 2 is not explicitly dealt a ``hand,'' we can view player 2 as always being dealt a medium-strength hand that wins against a losing hand and loses to a winning hand.) Both players have initial chip stacks of size $n$, and they both ante \$0.50 (creating an initial pot of \$1). P1 is allowed to bet any integral amount $x \in [0,n]$ (a bet of 0 is called a \emph{check}).\footnote{In the original formulation of the no-limit clairvoyance game~\cite{Ankenman06:Mathematics} player 1 is allowed to bet any real value in $[0,n]$, making the game a continuous game, since player 1's pure strategy space is infinite. For simplicity we consider the discrete game where player 1 is restricted to only betting integer values, though much of our equilibrium analysis will still apply for the continuous version as well.} Then P2 is allowed to call or fold (but not raise). This game clearly falls into the class of one-step extensive-form imperfect-information games. The game is small enough that its solution can be computed analytically (even for the continuous version)~\cite{Ankenman06:Mathematics}.

\begin{itemize}
\item P1 bets $n$ with prob. 1 with a winning hand.
\item P1 bets $n$ with prob. $\frac{n}{1+n}$ with a losing hand (and checks otherwise).
\item For all $x \in (0,n],$ P2 calls a bet of size $x$ with prob. $\frac{1}{1+x}$.
\end{itemize}

It was shown by Ankenman and Chen~\cite{Ankenman06:Mathematics} that this strategy profile constitutes a Nash equilibrium. (They also show that these frequencies are optimal in many other poker variants.) Here is a sketch of that argument. 

\begin{proposition}
\label{pr:clair_ne}
The strategy profile presented above is a Nash equilibrium of the no-limit clairvoyance game.
\end{proposition}

\begin{proof}
Player 2 must call a bet of size $x$ with probability $\frac{1}{1+x}$ in order to make player 1 indifferent between betting $x$ and checking with a losing hand. For a given $x,$ player 1 must bluff $\frac{x}{1+x}$ as often as he value bets for player 2 to be indifferent between calling and folding. Given these values, the expected payoff to player 1 of betting size $x$ is $v(x) = \frac{x}{2(1+x)}.$ This function is monotonically increasing, and therefore player 1 will maximize his payoff with $x = n$.
\end{proof}

Despite the simplicity of this game, the solution has been used in order to interpret bet sizes for the opponent that fall outside an abstracted game model by many strong agents for full no-limit Texas hold 'em~\cite{Ganzfried13:Action,Ganzfried17:Reflections,Jackson13:Slumbot}. Thus, its solution still captures important aspects of realistic forms of poker played competitively.

It turns out that player 2 does not need to call a bet of size $x \neq n$ with exact probability $\frac{1}{1+x}$: he need only not call with such an extreme probability that player 1 has an incentive to change his bet size to $x$ (with either a winning or losing hand). In particular, it can be shown that player 2 need only call a bet of size $x$ with any probability (which can be different for different values of $x$) in the interval $\left[\frac{1}{1+x}, \min \left\{\frac{n}{x(1+n)},1\right\} \right]$ in order to remain in equilibrium.\footnote{Note that if $x$ is required to be integral then we always have $x \geq \frac{n}{1+n}$, and $\min \left\{\frac{n}{x(1+n)},1\right\}  = \frac{n}{x(1+n)}$. However, our solution also holds for the continuous game.}

\begin{proposition}
A strategy profile $\sigma^*$ in the no-limit clairvoyance game is a Nash equilibrium if and only if under $\sigma^*$ player 1 bets $n$ with probability 1 with a winning hand and bets $n$ with probability $\frac{n}{1+n}$ with a losing hand, and for all $x \in (0,n]$ player 2 calls vs. a bet of size $x$ with probability in the interval $\left[\frac{1}{1+x}, \min \left\{\frac{n}{x(1+n)},1\right\} \right].$
\end{proposition}

\begin{proof}
We have already argued that the Nash equilibrium strategy for player 1 is unique. Note that for $x \in (0,n)$, player 1 will never actually bet $x$ in equilibrium, and player 2 must call with probability that ensures that player 1 will not want to deviate and bet $x$. Suppose player 2 calls with probability $p' < \frac{1}{1+x}$, for some $x \in (0,n).$ If player 1 bets $x$ with a losing hand, expected payoff is 
\begin{eqnarray*}
&&p'(-0.5-x) + (1-p')(0.5) \\
&= &p'(-1 - x) + 0.5\\
&> &\frac{-1 - x}{1+x} + 0.5\\
&= &-0.5\\
\end{eqnarray*}
If instead player 1 checks with a losing hand, the expected payoff is -0.5. So player 1 will strictly prefer to bet $x$ than to check, and will have incentive to deviate from his equilibrium strategy.

Now suppose player 2 calls with probability $p' > \min \left\{\frac{n}{x(1+n)},1\right\}$, for some $x \in (0,n).$ First suppose that $\frac{n}{x(1+n)} \leq 1.$ If player 1 bets $x$ with a winning hand, his expected payoff is
\begin{eqnarray*}
&&p'(0.5+x) + (1-p')(0.5) \\
& = &0.5p' + p'x + 0.5 - 0.5p'\\
& > &\frac{xn}{x(1+n)} + 0.5\\
& = &\frac{1.5n + 0.5}{1+n}\\
\end{eqnarray*}
If instead player 1 bets $n$, the expected payoff is
\begin{eqnarray*}
&&\frac{1}{1+n}(0.5+n) + \left(1-\frac{1}{1+n}\right)(0.5) \\
& = &\frac{1.5n + 0.5}{1+n} \\
\end{eqnarray*}
So player 1 will strictly prefer to bet $x$ than to bet $n$, and will have incentive to deviate from his equilibrium strategy.

Our analysis so far has shown that for $\frac{n}{x(1+n)} \leq 1,$ the strategy profile is a Nash equilibrium if and only if player 2 calls a bet of size $x$ with probability in the interval $\left[\frac{1}{1+x}, \frac{n}{x(1+n)} \right]$ for all $x \in (0,n].$

Now suppose that $\frac{n}{x(1+n)} > 1.$ Suppose that player 2 calls a bet of $x$ with probability 1. If player 1 bets $x$, his expected payoff is
\begin{eqnarray*}
& &1(0.5+x) + 0(0.5) \\
& = &x + 0.5 \\
& < &\frac{n}{1+n} + 0.5 \\
& = &\frac{1.5n + 0.5}{1+n}\\
\end{eqnarray*}
If instead player 1 bets $n$, the expected payoff as shown above is $\frac{1.5n + 0.5}{1+n}.$
So player 1 will not have incentive to deviate and bet $x$ instead of $n.$ 
\end{proof}

So we have shown that player 2 has infinitely many Nash equilibrium strategies that differ in their frequencies of calling vs. ``suboptimal'' bet sizes of player 1. Which of these strategies should we play when we encounter an opponent who bets a suboptimal size? One argument for calling with probability at the lower bound of the interval---$\frac{1}{1+x}$---is as follows (note that the previously-computed equilibrium strategy uses this value~\cite{Ankenman06:Mathematics}). If the opponent bets $x$ as opposed to the optimal  size of $n$ that he should bet in equilibrium, then a reasonable deduction is that he isn't even aware that $n$ would have been the optimal size, and believes that $x$ is optimal. Therefore, it would make sense to play a strategy that is an equilibrium in the game where the opponent is restricted to only betting $x$ (or to betting 0, i.e., checking). Doing so would correspond to calling a bet of $x$ with probability $\frac{1}{1+x}.$ The other equilibria pay more heed to the concern that the opponent could exploit us by deviating to bet $x$ instead of $n$; but we need not be as concerned about this possibility, since a rational opponent who knew to bet $n$ would not bet $x$.

One could use a similar argument to defend calling with probability at the upper bound of the interval---$\min \left\{\frac{n}{x(1+n)},1\right\}$. If the opponent somehow knew that betting $n$ was part of an optimal strategy but did not know that checking was, then perhaps we should follow an equilibrium of the game where the opponent is restricted to only betting $x$ or $n$, in which case our calling frequency should focus on dissuading the opponent from betting $x$ with a winning hand instead of $n.$ 

The first argument seems much more natural than the second, as it seems much more reasonable that a human is aware they should check sometimes with weak hands, but may have trouble computing that $n$ is the optimal size and guess that it is $x.$ However, both arguments could be appropriate depending on assumptions about the reasoning process of the opponent. The entire point of Nash equilibrium as a prescriptive solution concept is that we do not have any additional information about the players' reasoning process, so will opt to assume that all players are fully rational. If any additional information is available---such as historical data (either from our specific opponents' play or from a larger population of players), observations of play from the current match, a prior distribution, or any other model of the reasoning mechanism of the opponents---then we should clearly utilize this information and not simply follow a Nash equilibrium. Without any such additional information, it does not seem clear whether we should call with the lower bound probability, upper bound probability, or a value in the middle of the interval. The point of the equilibrium refinements we have considered is exactly to help us select between equilibria in a theoretically principled way in the absence of any additional information that could be used to model the specific opponents. 

For the remainder of our analysis we will restrict our attention to the no-limit clairvoyance game with $n = 2,$ with extensive-form game tree given by Figure~\ref{fi:clairvoyance}.

\begin{figure}[!ht]
\centering
\includegraphics[scale = 0.6]{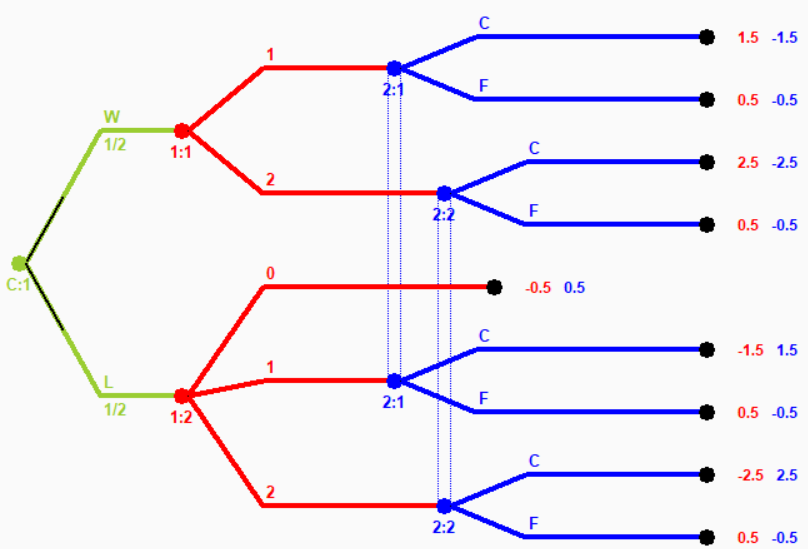}
\caption{No-limit clairvoyance game with $n = 2$.}
\label{fi:clairvoyance}
\end{figure}

According to our above analysis, the unique Nash equilibrium strategy for player 1 is to bet 2 with probability 1 with a winning hand, to bet 2 with probability $\frac{2}{3}$ with a losing hand, and to check with probability $\frac{1}{3}$ with a losing hand. The Nash equilibrium strategies for player 2 are to call a bet of 2 with probability $\frac{1}{3}$, and to call a bet of 1 with probability in the interval $\left[\frac{1}{2},\frac{2}{3}\right].$ As we show below, every Nash equilibrium of this game is
an EFTHPE, a QPE, an OSQPE with player 2 as the machine
player, and an extensive-form proper equilibrium.
Thus, all four refinements retain the full interval
$\left[\frac{1}{2},\frac{2}{3}\right]$ of calling
probabilities against a bet of 1.
By contrast, conditional on observing a bet of 1,
OPE uniquely selects the calling probability $\frac{5}{9}$.
The distinction is therefore that OPE selects a single
equilibrium from the continuum retained by the other
refinements.
The OPE value of $\frac{5}{9}$ corresponds to the solution
assuming that player 1 has bet 1 but that otherwise all
players are playing as rationally as possible.
Note also that the OPE does not simply correspond to the
average of the two interval boundaries, which would
be $\frac{7}{12}$.

\begin{proposition}
Every Nash equilibrium of the no-limit clairvoyance game with
$n=2$ is an extensive-form trembling-hand perfect equilibrium.
In particular, player 2 calls a bet of 2 with probability
$\frac{1}{3}$ and can call a bet of 1 with any probability
$\alpha\in[\frac{1}{2},\frac{2}{3}]$.
\end{proposition}

\begin{proof}
Fix $\alpha\in[\frac{1}{2},\frac{2}{3}]$, and let $\mu_\alpha$
denote the corresponding Nash equilibrium profile.
For $0<\epsilon<\frac{1}{6}$, define a perturbed game
$G_\epsilon$ by imposing a lower bound of $2\epsilon$ on
player 1's probability of betting 1 with a winning hand,
and a lower bound of $\epsilon$ on every other action
probability of both players.

Consider the following strategy for player 1:
\[
\begin{array}{c|ccc}
 & \text{Check} & \text{Bet 1} & \text{Bet 2}\\ \hline
\text{Winning hand} & \epsilon & 2\epsilon & 1-3\epsilon\\
\text{Losing hand} & \frac{1}{3}+\epsilon & \epsilon
                  & \frac{2}{3}-2\epsilon
\end{array}
\]
Let player 2 call a bet of 1 with probability $\alpha$
and a bet of 2 with probability $\frac{1}{3}$.
These strategies satisfy all lower bounds.

Against this player-2 strategy, player 1's payoffs with
a winning hand from checking, betting 1, and betting 2 are,
respectively,
\[
\frac{1}{2},\qquad
\frac{1}{2}+\alpha,\qquad
\frac{7}{6}.
\]
Since $\alpha\leq\frac{2}{3}$, betting 2 is optimal.
The proposed strategy assigns the required minimum probabilities
to the other two actions and is therefore a best response
with a winning hand.

With a losing hand, the corresponding payoffs are
\[
-\frac{1}{2},\qquad
\frac{1}{2}-2\alpha,\qquad
-\frac{1}{2}.
\]
Since $\alpha\geq\frac{1}{2}$, checking and betting 2 are
optimal. The proposed strategy assigns the required minimum
probability to betting 1 and is therefore also a best response
with a losing hand.

After observing a bet of 1, player 2's posterior probability
of a winning hand is
\[
\frac{2\epsilon}{2\epsilon+\epsilon}=\frac{2}{3}.
\]
After observing a bet of 2, it is
\[
\frac{1-3\epsilon}
     {(1-3\epsilon)+(\frac{2}{3}-2\epsilon)}
=\frac{3}{5}.
\]
Consequently, player 2's expected payoff from calling is
$-\frac{1}{2}$ after either bet, equal to the payoff from
folding. Thus player 2 is also best responding.

The proposed profile is therefore a Nash equilibrium of
$G_\epsilon$. As $\epsilon\rightarrow 0$, all lower bounds
vanish and the profile converges to $\mu_\alpha$.
Since $\alpha$ was arbitrary, every Nash equilibrium of
this game is EFTHPE.
\end{proof}

\begin{corollary}
Every Nash equilibrium of the no-limit clairvoyance game with
$n=2$ is a quasi-perfect equilibrium and, with player 2 as the
machine player, a one-sided quasi-perfect equilibrium.
For both concepts, player 2 calls a bet of 2 with probability
$\frac{1}{3}$ and can call a bet of 1 with any probability
$\alpha\in[\frac{1}{2},\frac{2}{3}]$.
\end{corollary}

\begin{proof}
Since neither player acts more than once along any path,
EFTHPE and QPE coincide in this game: there are no future
mistakes by the same player to distinguish the two concepts.
Furthermore, in the perturbation sequence constructed above,
player 2 is an unrestricted best response to player 1's
perturbed strategy for every $\epsilon$. Thus the same
sequence also establishes one-sided quasi-perfection,
with player 2 as the machine player.
\end{proof}

\begin{proposition}
Every Nash equilibrium of the no-limit clairvoyance game
with $n=2$ is an extensive-form proper equilibrium.
\end{proposition}

\begin{proof}
Fix $\alpha\in[\frac{1}{2},\frac{2}{3}]$.
For $0<t<\frac{1}{8}$, let $w_t=1-2t-t^2$ and consider
the following completely mixed strategy for player 1:
\[
\begin{array}{c|ccc}
 & \text{Check} & \text{Bet 1} & \text{Bet 2}\\ \hline
\text{Winning hand} & t^2 & 2t & w_t\\
\text{Losing hand} & 1-t-\frac{2}{3}w_t & t
                  & \frac{2}{3}w_t
\end{array}
\]
Let player 2 call a bet of 1 with probability $\alpha$
and a bet of 2 with probability $\frac{1}{3}$.

As in the preceding EFTHPE proof, player 1's payoffs
with a winning hand satisfy
\[
u_1(\text{Bet 2})\geq u_1(\text{Bet 1})
>u_1(\text{Check}),
\]
while with a losing hand they satisfy
\[
u_1(\text{Check})=u_1(\text{Bet 2})
\geq u_1(\text{Bet 1}).
\]
Since $w_t>\frac{1}{2}$ and
$1-t-\frac{2}{3}w_t>\frac{1}{3}$, the probability
ratios satisfy
\[
\frac{t^2}{2t}\leq 4t,\qquad
\frac{2t}{w_t}\leq 4t,\qquad
\frac{t^2}{w_t}\leq 4t,
\]
and
\[
\frac{t}{1-t-\frac{2}{3}w_t}\leq 4t,\qquad
\frac{t}{\frac{2}{3}w_t}\leq 4t.
\]
Thus, at each of player 1's information sets, every
strictly worse action is played with probability at
most $4t$ times that of every better action.
These inequalities also cover the endpoints of the
interval, where some actions become tied.

Player 2's posterior probability that player 1 has
a winning hand is $\frac{2}{3}$ after a bet of 1 and
$\frac{3}{5}$ after a bet of 2. Hence player 2 is
indifferent between calling and folding at both
information sets, so the properness conditions for
player 2 hold automatically.

The profile is therefore extensive-form $4t$-proper.
As $t\rightarrow 0$, it converges to the Nash equilibrium
in which player 2 calls a bet of 1 with probability
$\alpha$. Since $\alpha$ was arbitrary, every Nash
equilibrium of this game is extensive-form proper.
\end{proof}

\begin{proposition}
In the no-limit clairvoyance game with $n = 2$, the unique observable perfect equilibrium strategy for player 2 is to call vs. a bet of 1 with probability $\frac{5}{9}.$
\end{proposition}

\begin{proof}
Let $\mu$ denote the Nash equilibrium strategy profile where player 2 calls vs. a bet of 1 with probability $\frac{5}{9}.$ Consider the game $G_{\epsilon}$ where the sum of the probability that player 1 bets 1 with a winning hand and with a losing hand is at least $\epsilon.$ Consider the following strategy $\sigma^*_{\epsilon,1}$ for player 1. With a winning hand player 1 bets 2 with probability $1-\frac{2\epsilon}{3}$, and bets 1 with probability $\frac{2\epsilon}{3}$. With a losing hand player 1 bets 2 with probability $\frac{2(3-2\epsilon)}{9}$, bets 1 with probability $\frac{\epsilon}{3}$ and bets 0 with remaining probability $1 -  \frac{2(3-2\epsilon)}{9} - \frac{\epsilon}{3} = \frac{3+\epsilon}{9}.$ Consider whether player 1 can profitably deviate to strategy $\phi$ from this strategy when player 2 follows $\mu.$ First, we calculate expected payoff for player 1 of playing $\sigma^*_{\epsilon,1}$ against $\mu$:

\begin{eqnarray*}
&&\frac{1}{2}\left(\left(1-\frac{2\epsilon}{3}\right)\left(\frac{1}{3} (2.5) + \frac{2}{3}(0.5)\right) + \left(\frac{2 \epsilon}{3}\right)\left(\frac{5}{9}(1.5) + \frac{4}{9}(0.5)\right)\right)\\
&+ &\frac{1}{2}\left(\frac{2(3-2\epsilon)}{9}\left(\frac{1}{3}(-2.5) + \frac{2}{3}(0.5)\right) + \frac{\epsilon}{3} \left(\frac{5}{9}(-1.5) + \frac{4}{9}(0.5)\right) + \frac{3+\epsilon}{9} (-0.5)\right)\\
&= &\frac{1}{2}\left(\left(1-\frac{2\epsilon}{3}\right)\left(\frac{7}{6}\right) + \left(\frac{2 \epsilon}{3}\right)\left(\frac{19}{18}\right)\right)
+ \frac{1}{2}\left(\frac{2(3-2\epsilon)}{9}\left(-0.5\right) + \frac{\epsilon}{3} \left(-\frac{11}{18}\right) + \left(\frac{3+\epsilon}{9} (-0.5)\right)\right)\\
&= &\frac{1}{2}\left(-\frac{2\epsilon}{27} + \frac{7}{6}\right) + \frac{1}{2}\left(-\frac{1}{3} + \frac{2\epsilon}{9} - \frac{11 \epsilon}{54} - \frac{1}{6} - \frac{\epsilon}{18}\right)\\
&= &\frac{1}{3} -\frac{\epsilon}{18}\\
\end{eqnarray*}

Against $\mu$, checking with a winning hand yields
$\frac{1}{2}$, while betting 2 yields $\frac{7}{6}$.
Moving probability from checking to betting 2 preserves
the constraint on betting 1. Thus, it suffices to consider
deviations $\phi$ that never check with a winning hand.
For such a strategy, expected payoff is

\begin{eqnarray*}
&&\frac{1}{2}\left((1-\phi(W,1))\left(\frac{7}{6}\right) + \phi(W,1)\left(\frac{19}{18}\right)\right)\\ 
&+ &\frac{1}{2}\left(\phi(L,2)\left(-0.5\right) + \phi(L,1) \left(-\frac{11}{18}\right) + (1-\phi(L,2)-\phi(L,1)) (-0.5)\right)\\
&= &\frac{1}{2}\left(\frac{7}{6} - \frac{\phi(W,1)}{9}\right) + \frac{1}{2} \left(-0.5 - \frac{\phi(L,1)}{9}\right)\\
&= &\frac{1}{3} - \frac{\phi(W,1)}{18} - \frac{\phi(L,1)}{18}\\
\end{eqnarray*}

Given the constraint that $\phi(W,1) + \phi(L,1) \geq \epsilon$, this is maximized at
$\phi(W,1) + \phi(L,1) = \epsilon$, producing expected payoff $\frac{1}{3} - \frac{\epsilon}{18}.$
Since this is the same expected payoff as playing $\sigma^*_{\epsilon,1}$, player 1 cannot profitably deviate from $\sigma^*_{\epsilon,1}$.

Now consider player 2's best response. After observing a bet of 1, player 2's posterior probability that player 1 has a winning hand is
\begin{equation*}
\frac{\frac{2\epsilon}{3}}{\frac{2\epsilon}{3}+\frac{\epsilon}{3}}
=
\frac{2}{3}.
\end{equation*}
So player 2's payoff from calling is
\begin{equation*}
\frac{2}{3}(-1.5)+\frac{1}{3}(1.5)=-0.5,
\end{equation*}
which is equal to the payoff from folding. Thus player 2 is indifferent between calling and folding vs. a bet of 1, and can call with probability $\frac{5}{9}$.

After observing a bet of 2, player 2's posterior probability that player 1 has a winning hand is
\begin{equation*}
\frac{1-\frac{2\epsilon}{3}}{\left(1-\frac{2\epsilon}{3}\right)+\frac{2(3-2\epsilon)}{9}}
=
\frac{3}{5}.
\end{equation*}
So player 2's payoff from calling is
\begin{equation*}
\frac{3}{5}(-2.5)+\frac{2}{5}(2.5)=-0.5,
\end{equation*}
which is again equal to the payoff from folding. Thus player 2 is also indifferent between calling and folding vs. a bet of 2, and can call with probability $\frac{1}{3}$.

So the strategy profile where player 1 follows $\sigma^*_{\epsilon,1}$ and player 2 follows $\mu$ is a Nash equilibrium of $G_{\epsilon}.$ By taking the limit as $\epsilon \rightarrow 0$ it follows that $\mu$ is an observable perfect equilibrium of the original game $G.$

Now let $\mu_{\alpha}$ denote the Nash equilibrium strategy profile where player 2 calls vs. a bet of 1 with probability $\alpha$, where $\frac{1}{2} \leq \alpha \leq \frac{2}{3}$ and $\alpha \neq \frac{5}{9}.$ Suppose, toward a contradiction, that $\mu_{\alpha}$ is observable perfect. Then there exists a sequence $\epsilon_k \rightarrow 0$ and a sequence of Nash equilibria of $G_{\epsilon_k}$ converging to $\mu_{\alpha}.$ Let $\alpha_k$ denote player 2's probability of calling vs. a bet of 1 in the equilibrium of $G_{\epsilon_k}.$ Then $\alpha_k \rightarrow \alpha.$

For all sufficiently large $k$, player 1 checks and
bets 2 with positive probability with a losing hand,
since these probabilities converge to $\frac{1}{3}$
and $\frac{2}{3}$, respectively.
Moving probability between these two actions does not
affect the constraint on betting 1. Thus, in equilibrium,
these two actions must have equal payoffs.
Writing $\beta_k$ for player 2's probability of calling
a bet of 2, this gives
\[
-\frac{1}{2}=\frac{1}{2}-3\beta_k,
\]
and hence $\beta_k=\frac{1}{3}$.

Against this player-2 strategy, checking with a winning
hand yields $\frac{1}{2}$, while betting 2 yields
$\frac{7}{6}$. Moving probability from checking to
betting 2 also preserves the constraint on betting 1.
Therefore, every best response assigns zero probability
to checking with a winning hand.

Against a player-2 strategy that calls vs. a bet of 1 with probability $\alpha_k$, player 1's expected payoff from a strategy $\phi$ is

\begin{eqnarray*}
&&\frac{1}{2}\left((1-\phi(W,1))\left(\frac{1}{3} (2.5) + \frac{2}{3}(0.5)\right) +  \phi(W,1)\left(
\alpha_k(1.5) + (1-\alpha_k)(0.5)\right)\right)\\ 
&+ &\frac{1}{2}\left(\phi(L,2)\left(\frac{1}{3}(-2.5) + \frac{2}{3}(0.5)\right) +  \phi(L,1) \left(\alpha_k(-1.5) + (1-\alpha_k)(0.5)\right)\right)\\ 
&+ &\frac{1}{2}\left((1-\phi(L,2)-\phi(L,1)) (-0.5)\right)\\
&= &\frac{1}{2} \left(\phi(W,1)\left(\alpha_k - \frac{2}{3}\right) + \phi(L,1)(-2\alpha_k +1) + \frac{2}{3}\right).
\end{eqnarray*}

We want to maximize this subject to
\begin{equation*}
\phi(W,1)+\phi(L,1)\geq \epsilon_k.
\end{equation*}

First suppose $\alpha < \frac{5}{9}.$ Then for all sufficiently large $k$, we have $\alpha_k < \frac{5}{9}.$ For such $k$, the coefficient on $\phi(L,1)$ is strictly larger than the coefficient on $\phi(W,1)$, since
\begin{equation*}
(-2\alpha_k+1)-\left(\alpha_k-\frac{2}{3}\right)
=
\frac{5}{3}-3\alpha_k
>
0.
\end{equation*}
Therefore, in any best response in \(G_{\epsilon_k}\), player 1 sets
\[
\phi(W,1)=0
\]
and chooses \(\phi(L,1)\) so that
\[
\phi(L,1)\ge \epsilon_k.
\]
Against such a player-1 strategy, after observing a bet of 1, player 2 strictly prefers calling to folding. Hence in any Nash equilibrium of $G_{\epsilon_k}$ player 2 must call vs. a bet of 1 with probability 1. This contradicts $\alpha_k \rightarrow \alpha < \frac{5}{9}.$

Now suppose $\alpha > \frac{5}{9}.$ Then for all sufficiently large $k$, we have $\alpha_k > \frac{5}{9}.$ For such $k$, the coefficient on $\phi(W,1)$ is strictly larger than the coefficient on $\phi(L,1)$, since
\begin{equation*}
\left(\alpha_k-\frac{2}{3}\right)-(-2\alpha_k+1)
=
3\alpha_k-\frac{5}{3}
>
0.
\end{equation*}
Therefore, in any best response in \(G_{\epsilon_k}\), player 1 sets
\[
\phi(L,1)=0
\]
and chooses \(\phi(W,1)\) so that
\[
\phi(W,1)\ge \epsilon_k.
\]
Against such a player-1 strategy, after observing a bet of 1, player 2 strictly prefers folding to calling. Hence in any Nash equilibrium of $G_{\epsilon_k}$ player 2 must call vs. a bet of 1 with probability 0. This contradicts $\alpha_k \rightarrow \alpha > \frac{5}{9}.$

Therefore no sequence of Nash equilibria of the perturbed games can converge to $\mu_{\alpha}$ for any $\alpha \neq \frac{5}{9}.$ Hence the unique observable perfect equilibrium strategy for player 2 is to call vs. a bet of 1 with probability $\frac{5}{9}$.
\end{proof}

The OPE solution can be motivated intuitively by the following reasoning. (For this discussion, payoffs exclude the initial ante, so they are
$\frac{1}{2}$ higher than in the preceding calculations; payoff
differences are unchanged.)
First, note that if player 1 bets the optimal size of 2 with a winning hand against an equilibrium strategy (which calls a bet of 2 with probability $\frac{1}{3}$), expected payoff is 
$$\frac{1}{3} (3) + \frac{2}{3} (1) = \frac{5}{3}.$$
If player 1 bets 1 with a winning hand against a strategy that calls with probability $\alpha$, expected payoff is
$$\alpha (2) + (1-\alpha) (1) = \alpha + 1.$$
The difference between these two values is
$$\frac{5}{3} - (\alpha + 1) = \frac{2}{3} - \alpha.$$

If player 1 bets 2 with a losing hand against an equilibrium strategy, expected payoff is
$$\frac{1}{3}(-2) + \frac{2}{3}(1) = 0.$$
Similarly, if player 1 bets 0 with a losing hand against an equilibrium strategy, expected payoff is 0.
If player 1 bets 1 with a losing hand against a strategy that calls with probability $\alpha$, expected payoff is
$$\alpha (-1) + (1-\alpha) (1) = 1 - 2 \alpha.$$
The difference between these two values is
$$0 - (1 - 2 \alpha) = 2 \alpha - 1$$

These two differences are equal iff
$$\frac{2}{3} - \alpha = 2 \alpha - 1$$
$$3 \alpha = \frac{5}{3}$$
$$\alpha = \frac{5}{9}.$$

These payoff losses give OPE a maximin interpretation.
Since the game is zero-sum, each payoff loss for player 1
is an equal payoff gain for player 2, relative to player 1's
optimal play with the same hand. Among Nash equilibrium
responses, the calling probability that maximizes player 2's
worst-case gain over the two possible types of mistake solves
\[
\max_{\alpha\in[\frac{1}{2},\frac{2}{3}]}
\min\left\{
\frac{2}{3}-\alpha,\;
2\alpha-1
\right\}.
\]
The unique maximizer is $\alpha=\frac{5}{9}$, at which
both gains equal $\frac{1}{9}$.
Thus, OPE maximizes player 2's worst-case gain from
the opponent's observed mistake.

\section{Conclusion}
\label{se:conc}
We presented a new solution concept for sequential imperfect-information games called observable perfect equilibrium that captures the assumption that all players are playing as rationally as possible given the fact that some players have taken observable suboptimal actions. We believe that this is more compelling than other solution concepts that assume that one or all players make certain types of mistakes for all other actions including those that have not been observed. We showed that every observable perfect equilibrium is a Nash equilibrium, which implies that observable perfect equilibrium is a refinement of Nash equilibrium. We also showed that observable perfect equilibrium is always guaranteed to exist. We showed that an OPE can be computed in polynomial time in two-player zero-sum games based on repeatedly solving a linear program formulation. The OPE formulation shares a useful computational property
with OSQPE: the perturbation parameter appears only in
the objective, leaving the feasible region independent
of its magnitude. Its simpler perturbation structure
may offer additional computational benefits.

We demonstrated that observable perfect equilibrium selects
a unique equilibrium in a simplified no-limit poker game
where EFTHPE, QPE, OSQPE, and extensive-form proper
equilibrium retain the full set of Nash equilibria.
While we only considered a simplified game called the no-limit clairvoyance game, this game encodes several elements of the complexity of full no-limit Texas hold 'em, and in fact conclusions from this game have been incorporated into some of the strongest agents for no-limit Texas hold 'em. So we expect our analysis to extend to significantly more complex settings than the example considered. 
We think that observable perfect equilibrium captures the theoretically-correct solution concept for sequential imperfect-information games, and may offer practical computational advantages.

Future work should explore the scalability of the algorithm we have presented. While algorithms based on solving linear programs run in polynomial time, they often run into memory and speed issues that prevent them from being competitive with algorithms such as counterfactual regret minimization and fictitious play on extremely large games. Perhaps these algorithms can be modified in such a way that they are guaranteed to converge to an observable perfect equilibrium in two-player zero-sum games (currently they are just guaranteed to converge to Nash equilibrium and not to any specific refinement).

While the focus of this paper has been on two-player games, and two-player zero-sum games specifically, the concept of observable perfect equilibrium is generally applicable to multiplayer games and non-zero-sum games as well. It has been argued that one-sided quasi-perfect equilibrium is inappropriate for non-zero-sum (and multi-player) games and should not be used~\cite{Farina21:Equilibrium}. By contrast, we see no reason why observable perfect equilibrium cannot be applied to these games with theoretical and computational advantages analogous to those in the two-player zero-sum setting. In the future we would like to study the applicability of observable perfect equilibrium to these settings and develop new algorithms for its computation.

\bibliographystyle{plain} 
\bibliography{C://FromBackup/Research/refs/dairefs}
\end{document}